\documentclass[12pt]{article}

\usepackage{amsmath,amssymb,amsthm}
\usepackage{braket}
\usepackage{geometry}
\usepackage{amsmath}
\DeclareMathOperator{\Tr}{Tr}

\title{\bf {No-Signalling Fixes the Hilbert-Space Inner Product}}
\author{Arun K Pati, \\
Centre for Quantum Technology, \\
KIIT University, Bhubanseswar, India}
\date{}

\newtheorem{theorem}{Theorem}

\begin{document}

\maketitle

\begin{abstract}
We investigate whether the inner product structure of quantum mechanics can be modified without violating fundamental physical principles. We consider a generalized inner product defined by a positive operator 
and assume local unitary dynamics, existence of entangled states and the no-signalling principle. We show that any nontrivial choice of inner product different from standard one inevitably leads to superluminal signalling, in contradiction with relativistic causality. 
Therefore, the standard Hilbert-space inner product is uniquely enforced by no-signalling. 
\end{abstract}

\section{Introduction}

The no-signalling principle occupies a central place in the conceptual and operational foundations of quantum theory. It asserts that operations performed locally on one part of a composite system cannot be used to transmit information instantaneously to another spatially separated part. While originally motivated by relativistic causality, no-signalling has emerged as a far deeper structural constraint: it severely restricts the possible forms of state spaces, probability rules, and dynamical transformations that any consistent physical theory may admit. In quantum mechanics \cite{holevo}, no-signalling is not an added postulate but rather a delicate consequence of the interplay between the Hilbert-space structure, the tensor product composition of systems, and the Born rule.

From an information-theoretic perspective, no-signalling guarantees the autonomy of local laboratories. It ensures that reduced states and local measurement statistics are well defined independently of distant choices, thereby making quantum theory compatible with operational notions of locality and experimental reproducibility. At the same time, quantum theory allows nonlocal correlations, as evidenced by Bell inequality violations \cite{bell,Tsi}, which saturate the no-signalling boundary without crossing it \cite{kai}. This coexistence of nonlocality and no-signalling is one of the most distinctive and nontrivial features of quantum mechanics, setting it apart from both classical theories and more general probabilistic models \cite{dago}.

A growing body of work has shown that even slight modifications to the mathematical structure of quantum theory can jeopardize no-signalling. Nonlinear evolutions, altered probability rules, particularly in composite and entangled systems—generically enable superluminal communication \cite{weinberg,gisin,poli,wein} unless stringent consistency conditions are imposed. These results suggest that no-signalling acts as a powerful organizing principle, tightly constraining the admissible kinematics and dynamics of physical theories. 

The importance of no-signalling extends beyond foundational concerns. It plays a crucial role in quantum information processing, cryptography, and metrology, where security, composability, and resource accounting rely implicitly on the impossibility of signalling through entanglement alone. Moreover, in emerging contexts such as modified quantum theories, effective descriptions of open systems, and attempts to unify quantum mechanics with gravity, preserving no-signalling provides a stringent consistency check on proposed extensions.

In this work, we examine the role of no-signalling as a structural principle underlying quantum theory. We analyze how it constrains the form of inner products, probability assignments, and composite-system structures, and we demonstrate that violations of these constraints generically lead to signalling in entangled systems. Our results clarify why the standard quantum formalism is remarkably rigid, and they highlight no-signalling not merely as a relativistic requirement, but as a cornerstone of the internal logical consistency of quantum mechanics itself.

The mathematical structure of quantum mechanics is rigidly constrained by physical principles, yet the necessity of the standard inner product remains a subtle foundational question. While the Hilbert space formalism is often assumed axiomatically, one might ask whether alternative geometries—mediated by a non-trivial metric operator could consistently describe physical reality. In this Letter, we demonstrate that such a modification is not merely a choice of convention but a direct violation of causality. By analyzing the interplay between entanglement, measurement, and the generalized normalization condition, we establish that any deviation from the standard metric inevitably destroys the linearity required to preserve the indistinguishability of ensembles. Consequently, the standard Hilbert space inner product is identified not as a mathematical convenience, but as the unique one that safeguards the no-signaling principle of special relativity.

Quantum mechanics and special relativity coexist peacefully only because of the precise linear structure of the former. A central pillar of this structure is the inner product, which defines probability amplitudes and ensures the unitary conservation of information. 
We demonstrate that replacing the standard inner product with a generalized form $\langle \psi | \phi \rangle_A:= \bra{\psi}A\ket{\psi}$ with $A >0$ leads to an immediate conflict with relativistic causality. Specifically, we show that any non-trivial metric $A \neq \mathbb{I}$ induces a transformation that renders orthogonal measurement ensembles distinguishable. This allows for superluminal signaling using shared entanglement, establishing that the standard quantum inner product is the unique metric consistent with a relativistic universe.

This result is important because it shows that the Hilbert-space inner product of quantum mechanics is not merely a mathematical convention, but is rigidly constrained by fundamental physical principles. Any attempt to modify the inner product—while keeping linear unitary dynamics and entanglement—inevitably introduces state-dependent normalization, leading to nonlinear evolution and superluminal signalling. Thus, relativistic causality and no-signalling uniquely select the standard inner product (up to scale), much as Lorentz invariance constrains spacetime geometry. The analysis clarifies why many proposed “generalized” or “deformed” quantum theories fail at a foundational level: even minimal departures from the standard inner product destabilize the operational equivalence of ensembles. In this sense, the Hilbert-space structure emerge not as independent axioms, but as consequences of deeper consistency requirements linking quantum mechanics and relativistic causality.


Let $\mathcal{H}$ be a Hilbert space with the standard inner product denoted by $\langle \cdot | \cdot \rangle$. Let $A$ be a linear operator on $\mathcal{H}$ such that $A$ is positive definite ($A > 0$). We define the new inner product, denoted $\langle \cdot, \cdot \rangle_A$, as a map $\mathcal{H} \times \mathcal{H} \to \mathbb{C}$ given by:

\begin{align}  
\langle \phi, \psi \rangle_A := \langle \phi | A | \psi \rangle
\end{align}


To prove $\langle \phi , \psi \rangle_A$ is a valid inner product, we must verify the three axioms of a Hilbert space inner product: Conjugate Symmetry, Linearity, and Positive Definiteness.\\

\noindent
Axiom 1: Conjugate Symmetry- We need to show that $\langle \psi , \phi \rangle_A = \overline{\langle \phi, \psi \rangle_A}$.\\

\noindent
Proof: By definition, we have $\langle \psi, \phi \rangle_A = \langle \psi | A | \phi \rangle$.
Using the property of the adjoint $A^\dagger$ in the standard inner product, we have $\langle \psi | A | \phi \rangle = \langle A^\dagger \psi | \phi \rangle = \overline{\langle \phi | A^\dagger \psi \rangle}$.
Since $A$ is a positive operator ($A > 0$), and is Hermitian (self-adjoint), it holds that 
 $ \overline{\langle \phi | A | \psi \rangle} = \overline{\langle \phi, \psi \rangle_A}$. Thus, conjugate symmetry holds.\\

\noindent
Axiom 2: Linearity-We need to show linearity in the second argument: $\langle \phi, c_1 \psi_1 + c_2 \psi_2 \rangle_A = c_1 \langle \phi, \psi_1 \rangle_A + c_2 \langle \phi, \psi_2 \rangle_A$.\\

\noindent
Proof: Consider $\langle \phi, c_1 \psi_1 + c_2 \psi_2 \rangle_A = \langle \phi | A (c_1 |\psi_1\rangle + c_2 |\psi_2\rangle)$.
Since $A$ is a linear operator, above can be written as $ \langle \phi | (c_1 A |\psi_1\rangle + c_2 A |\psi_2\rangle)$.
Using the linearity of the standard inner product, we have 
$$\langle \phi, c_1 \psi_1 + c_2 \psi_2 \rangle_A
= c_1 \langle \phi, \psi_1 \rangle_A + c_2 \langle \phi, \psi_2 \rangle_A$$
Thus, linearity holds.\\

\noindent
Axiom 3: Positive Definiteness-We need to show that $\langle \psi, \psi \rangle_A \geq 0$, and that it equals $0$ if and only if $|\psi\rangle = 0.$\\

\noindent
Proof: Consider $$\langle \psi, \psi \rangle_A = \langle \psi | A | \psi \rangle$$. By the definition of a positive definite operator ($A > 0$), for any non-zero vector $|\psi\rangle \neq 0$, the expectation value $\langle \psi | A | \psi \rangle > 0$.
Therefore, $\langle \psi, \psi \rangle_A \geq 0$ for all $| \psi \rangle$.
If $\langle \psi, \psi \rangle_A = 0$, then $\langle \psi | A | \psi \rangle = 0$. Since $A$ is strictly positive definite, this is only possible if $|\psi\rangle = 0$. Thus, positive definiteness holds.

Since the map satisfies all three axioms, $\langle \phi | A | \psi \rangle$ defines a valid inner product on the space.
This is often called the skewed Inner Product in physics. 

\begin{theorem}[Non-standard inner products imply signalling]
\label{thm:inner-product-signalling}
Let the inner product on a Hilbert space ${\cal H}$ be modified to
\[
\langle \psi ,  \phi \rangle_A := \langle \psi | A | \phi \rangle,
\]
where $A>0$ is a fixed positive operator.  
If $A \neq I$, while allowing standard local unitaries and entanglement, then the no-signalling principle is violated.
\end{theorem}

\begin{proof}


Let $A>0$ be a fixed positive operator, not proportional to the identity.  
Probabilities are computed using the modified norm
\[
\|\psi\|_A^2 := \langle \psi | A | \psi \rangle .
\]
Physical pure states must satisfy the normalization condition
\[
\langle \psi | A | \psi \rangle = 1,
\]
which replaces the usual Born rule normalization.


Let $U$ be a unitary acting linearly as $\ket{\psi} \mapsto U \ket{\psi}$.
After evolution, the norm becomes $\| U \psi \|_A^2
= \langle \psi | U^\dagger A U | \psi \rangle $. If $A \neq c I$, then there exists at least one unitary $U$ such that
$U^\dagger A U \neq A$,
and therefore $\| U \psi \|_A^2 \neq \| \psi \|_A^2$.
Hence, unitary operations change probabilities.


To restore physical meaning, the post-unitary state must be renormalized:
\begin{align}
\ket{\psi} \longmapsto
\frac{U \ket{\psi}}
{\sqrt{\langle \psi | U^\dagger A U | \psi \rangle}} = \frac{1}{\sqrt{N(\psi)}} U \ket{\psi},
\end{align}
where $N(\psi) := \langle \psi | U^\dagger A U | \psi \rangle$.
The resulting evolution map is nonlinear, since in general
\begin{align}
N(\alpha \ket{\psi} + \beta \ket{\phi})
\neq
\alpha N(\psi) + \beta N(\phi).
\end{align}
Now, consider a maximally entangled state shared by Alice and Bob as given by 
\[
\ket{\Psi}_{AB}
=
\frac{1}{\sqrt{2}}
\bigl( \ket{0}_A \ket{0}_B + \ket{1}_A \ket{1}_B \bigr).
\]
Bob is spacelike separated from Alice. Bob prepares two operationally equivalent ensembles, i.e., he chooses one of two measurements: (i) Computational basis $\{\ket{0},\ket{1}\}$ and  (ii) Diagonal basis $\{\ket{+},\ket{-}\}$. 
In both cases, Alice’s reduced state is $\rho_A = \frac{I}{2}$.

 Let Alice applies a local unitary not commuting with $A$.
Due to the modified inner product, each pure state must be renormalized individually and  different ensembles evolve differently.
Because normalization depends on the state, the evolution does not respect convex linearity:
\begin{align}
\mathcal{N}\!\left( \sum_i p_i \ket{\psi_i}\!\bra{\psi_i} \right)
\neq
\sum_i p_i \mathcal{N}(\ket{\psi_i}\!\bra{\psi_i}),
\end{align}
where $\mathcal{N}$ is the effective nonlinear state-update map induced by
state-dependent renormalization under a modified inner product.
Hence, the $\{\ket{0},\ket{1}\}$ ensemble and the $\{\ket{+},\ket{-}\}$ ensemble evolve into different final density matrices, despite starting from the same $\rho_A$. Alice now measures her system and observes outcome statistics that depend on Bob’s choice of measurement basis. This dependence occurs without any classical communication. Thus, Bob can signal instantaneously to Alice.

Thus any modification of the inner product that makes normalization state-dependent introduces nonlinearity.  
Nonlinearity together with entanglement inevitably enables superluminal signalling. Therefore,
no-signalling implies that $A =  I$.
The standard Hilbert-space inner product is not a convention but is enforced by relativistic causality.
\end{proof}

\section{Example}
Here, we give a specific example for two-qubit entangled state.
Consider two parties, Alice and Bob, sharing a maximally entangled Bell state
\begin{align}
\ket{\Psi} = \frac{1}{\sqrt{2}} (\ket{0}_A \ket{0}_B + \ket{1}_A \ket{1}_B).
\end{align}
Bob is far away from Alice and performs a measurement on his qubit. Alice applies the Hadamard gate $H$ on her qubit and  
$ [H,A]\neq 0 $ for  $\lambda\neq 1 $.  Let us consider two measurement choices for Bob.

\subsubsection{Case 1: Diagonal basis $\{|0\rangle,|1\rangle\}$}
 First, consider the case where Bob measures his qubit in the computational basis.  
After Bob’s measurement, 
Alice’s ensemble is given by 
\[
\mathcal{E}_Z
=
\Bigl\{\tfrac12,|0\rangle;\ \tfrac12,|1\rangle\Bigr\}.
\]
Now, Alice applies $H$ and renormalizes. Starting from $|0\rangle$, we have 
$H|0\rangle = |+\rangle = \frac{|0\rangle+|1\rangle}{\sqrt{2}}$.
The normalization factor is
$\langle +|A|+\rangle = \frac{1}{2}(1+\lambda)$.
The renormalized state is given by 
$|+\rangle_A = \frac{|+\rangle}{\sqrt{(1+\lambda)/2}}$.
Similarly, starting from $|1\rangle$, we hve 
$H|1\rangle = |-\rangle = \frac{|0\rangle-|1\rangle}{\sqrt{2}}$.
The normalization factor is given by $\langle -|A|-\rangle = \frac{1}{2}(1+\lambda)$.
The renormalized state is now
$|-\rangle_A = \frac{|-\rangle}{\sqrt{(1+\lambda)/2}}$.
Therefore, the final density matrix is given by
\[
\rho_A^{(Z)}
=
\frac12 |+\rangle_A\!\langle+|
+
\frac12 |-\rangle_A\!\langle-|
=
\frac{1}{1+\lambda}
\begin{pmatrix}
1 & 0 \\
0 & 1
\end{pmatrix}.
\]

\subsubsection{Case 2: Diagonal basis $\{|+\rangle,|-\rangle\}$}
Next, consider the measurement by Bob in the $X$ basis.  After Bob’s measurement, 
 Alice’s ensemble is given by 
\[
\mathcal{E}_X
=
\Bigl\{\tfrac12,|+\rangle;\ \tfrac12,|-\rangle\Bigr\}.
\]

Now, Alice applies $H$ and renormalizes the states. Starting from $|+\rangle$, we have 
$H|+\rangle = |0\rangle$. The normalization factor is $\langle 0|A|0\rangle = 1$, so the state is unchanged.
Starting from $|-\rangle$, we have  $H|-\rangle = |1\rangle$. The normalization factor is given by 
$\langle 1|A|1\rangle = \lambda$. The renormalized state is now
$|1\rangle_A = \frac{|1\rangle}{\sqrt{\lambda}}$. Therefore, the
final density matrix is given by
\[
\rho_A^{(X)}
=
\frac12 |0\rangle\!\langle 0|
+
\frac12 \frac{|1\rangle\!\langle 1|}{\lambda}
=
\begin{pmatrix}
\frac12 & 0 \\
0 & \frac{1}{2\lambda}
\end{pmatrix}.
\]

Now, Alice measures the projector $M = |0\rangle\!\langle 0|$. The corresponding probability for the Case 1 is 
$P_Z(0) = \mathrm{Tr}\!\left(AM\rho_A^{(Z)}\right)
= \frac{1}{1+\lambda}$ and the probability for the Case 2  is 
$P_X(0) = \mathrm{Tr}\!\left(AM\rho_A^{(X)}\right) = \frac{1}{2}$. For $\lambda\neq 1$, $P_Z(0) \neq P_X(0)$.
Hence, Alice can distinguish Bob’s measurement choice using a single local measurement, implying superluminal signalling.

Thus, for Alice to be unable to distinguish Bob's choice, we must have $\rho_1 = \rho_2$. 
Unless $\lambda = 1$ (which means $A=I$), we find that $\rho_1 \neq \rho_2$.
If $A \neq I$, Alice can simply measure the population of her qubit. If she sees the distribution corresponding to $\rho_1$, she knows Bob measured in the Standard Basis. If she sees $\rho_2$, she knows he measured in the Diagonal Basis.
This can lead to super luminal signaling. To prevent this, nature enforces $A=I$, hence the standard Hilbert space inner product.

\section{Discussion and Conclusions}
We have shown that under the assumptions of linear unitary dynamics, operational consistency with entanglement, and no-signalling, the inner product structure of quantum mechanics is uniquely fixed (up to scale). Any attempt to define a modified inner product using a nontrivial positive operator $A$ introduces nonlinearity, and allows superluminal signalling. Hence, the 
standard Hilbert-space inner product is enforced by fundamental physical principles.

To understand why the non-standard inner product leads to signalling, let us  visualize the Bloch Sphere.
In standard Quantum Mechanics ($A=I$),  for a qubit, the state space is a perfect sphere. Rotations (unitary operations) move states on the surface of the Bloch sphere, keeping the distance from the origin (norm) constant.
In the modified inner product ($A \neq I$): The ``unit norm'' condition $\bra{\psi}A\ket{\psi}=1$ defines an ellipsoid 
(a squashed sphere). When Bob measures in different bases, he prepares mixtures that lie on different axes of this ellipsoid.
Because the shape is not a sphere, the ``center of mass'' of the mixture on the long axis is different from the center of mass of the mixture on the short axis. The fact that the standard traces are weird ($\neq 1$) is the mathematical symptom of the geometry being warped. This warping allows Alice to distinguish the ensembles, leading to the violation of no-signalling.

A well-studied instance closely related to our result involves PT-symmetric quantum mechanics \cite{carl} which is a non-Hermitian extension of 
standard QM. In  PT-symmetric quantum mechanics one uses a non-standard inner product to restore unitary evolution and norm preservation. Several works showed that 
if one naively uses conventional inner product structure in PT-symmetric evolutions for bipartite entangled systems, then entanglement can be changed under local operation \cite{akp}, the no-signalling principle can be violated \cite{pt}. Therefore, previous results can be now understood as special cases of this more general result.  To preserve no-signalling, one must adopt a new, consistent inner product (such as the CPT inner product) applied uniformly to both parties. Our result is not about arbitrary inner products per se, but they demonstrate the general principle:
Changing the inner product without ensuring its compatibility with the tensor-product structure and operational probability rule can lead to signalling. Any theory that lacks consistency will violate the no-signalling principle.

In future, it will be worth investing if the Born rule as the only consistent probability assignment also follows from the no-signalling. If answer comes in favor of this, then that can establishe no-signalling as a unifying physical principle underlying both the geometry and the probabilistic structure of quantum mechanics, sharpening and extending earlier reconstructions of quantum theory.

\section{Appendix}

\subsection{Trace Condition} 

One may notice that the standard trace $\Tr(\rho)$ is not 1. This is actually a feature, not a bug, of the modified theory we are analyzing. It confirms that we have fundamentally changed the geometry of the state space.
At first, it appears that the trace ``looks'' wrong, but it is actually correct in the new theory.
The ``new'' normalization rule will ensure that the density matrices are normalised indeed.

In standard Quantum Mechanics, the condition $Tr(\rho)=1$ comes from the requirement that probabilities sum to one, i.e., 
\[
\sum_i P_i = \sum \langle i | \psi \rangle \langle \psi | i \rangle  = \Tr(\ket{\psi}\bra{\psi}) = 1.
\]
However, in this hypothetical theory, we defined the inner product as $\bra{\psi}A\ket{\psi}$. Therefore, the condition that probabilities sum to one changes to:
\[
\Tr(A\rho) = 1.
\]
Let us check if our derived matrices satisfy this new physical consistency condition. For the Case 1, we have 
\[
\rho_1 = \begin{pmatrix} \frac{1}{1+\lambda} & 0 \\ 0 & \frac{1}{1+\lambda} \end{pmatrix}, 
\quad A = \begin{pmatrix} 1 & 0 \\ 0 & \lambda \end{pmatrix}.
\]
Now, we find the weighted trace. Since 
\[
A\rho_1 = \frac{1}{1+\lambda} \begin{pmatrix} 1  & 0 \\ 0 & \lambda \end{pmatrix},
\]
we have $\Tr(A\rho_1) = 1$ and $\rho_1$ is a valid, normalized state in this new theory. For the Case 2, we have  
\[
\rho_2 = \begin{pmatrix} \frac{1}{2} & 0 \\ 0 & \frac{1}{2\lambda} \end{pmatrix}, \quad A = \begin{pmatrix} 1 & 0 \\ 0 & \lambda \end{pmatrix}.
\]
To calculate the weighted trace, we have 
\[
A\rho_2 = \begin{pmatrix} \frac{1}{2} & 0 \\ 0 & \frac{1}{2} \end{pmatrix}.
\]
Hence, $\Tr(A\rho_2) = 1$ and $\rho_2$ is also a valid, normalized state in this new theory. However, $\rho_1$ and $\rho_2$ are different matrices, i.e., $ \rho_1 \neq \rho_2$.  In any linear physical theory, the statistics of a measurement for an observable $M$ are determined by the density matrix 
(e.g., $\langle M \rangle = Tr(AM\rho)$).
Since the matrices are different, there exists a measurement $M$ that gives different results for $\rho_1$ and $\rho_2$.
Alice simply performs this measurement. If she gets result $Z$, she knows Bob measured in the standard basis. 
If she gets result $X$, she knows Bob measured in the diagonal basis.
Since she can tell the difference instantly, this can violate the no-signalling condition.

\subsection{The Nonlinear State-Update Map }

We denote by $\mathcal{N}$ the \emph{effective nonlinear state-update map} induced by
state-dependent renormalization under a modified inner product.
To see why this map is non-linear, fix a positive operator $A > 0$ on a qubit Hilbert space $\mathcal{H} \cong \mathbb{C}^2$,
with $A \neq c\,\mathbb{I}$, for any constant $c>0$. The inner product is defined as
$\langle \psi | \phi \rangle_A := \langle \psi | A | \phi \rangle$,
and physical pure states satisfy the normalization condition $\langle \psi | A | \psi \rangle = 1$.

First, consider the case of pure state evolution and renormalization. Let $U \in SU(2)$ be a local unitary such that
$[U,A] \neq 0$. After applying $U$, the state $|\psi\rangle$ is no longer normalized with respect to
the $A$-inner product. The physically meaningful post-evolution state must therefore be
renormalized. We define the nonlinear map $\mathcal{N}$ acting on pure states as
\begin{equation}
\mathcal{N}\!\left( \ket{\psi}\!\bra{\psi} \right)
:=
\frac{U \ket{\psi}\!\bra{\psi} U^\dagger}
{\langle \psi | U^\dagger A U | \psi \rangle}.
\label{eq:N_pure}
\end{equation}
The denominator 
\begin{equation}
N(\psi) := \langle \psi | U^\dagger A U | \psi \rangle
\label{eq:normalization}
\end{equation}
is a state-dependent normalization factor. Because $[U,A]\neq 0$, this factor depends nontrivially on $|\psi\rangle$.

Next, we can define its extension to mixed states via ensembles representation.
Let a mixed state be represented operationally by an ensemble
\[
\rho = \sum_i p_i \ket{\psi_i}\!\bra{\psi_i},
\qquad
p_i \ge 0,\quad \sum_i p_i = 1.
\]
Operationally, Alice must apply the transformation to each pure state in the ensemble
and then re-average the outcomes. The physically relevant evolution rule is therefore
\begin{equation}
\mathcal{N}(\rho)
:=
\sum_i p_i
\frac{U \ket{\psi_i}\!\bra{\psi_i} U^\dagger}
{\langle \psi_i | U^\dagger A U | \psi_i \rangle}.
\label{eq:N_mixed}
\end{equation}
This ensemble-level rule follows directly from the requirement that each pure state
be individually normalized according to the modified inner product.
However, this map fails to be linear and brings ensemble dependence.

Now compare the following two expressions:

\begin{enumerate}
\item Applying $\mathcal{N}$ to the mixed state as a whole:
\[
\mathcal{N}\!\left( \sum_i p_i \ket{\psi_i}\!\bra{\psi_i} \right).
\]

\item Taking the mixture of individually transformed pure states:
\[
\sum_i p_i \, \mathcal{N}\!\left( \ket{\psi_i}\!\bra{\psi_i} \right).
\]
\end{enumerate}
Using Eqs.~\eqref{eq:N_pure}--\eqref{eq:N_mixed}, these expressions are \emph{not equal} because
the normalization factors
\[
\frac{1}{\langle \psi_i | U^\dagger A U | \psi_i \rangle}
\]
depend explicitly on the index $i$. Hence,
\begin{equation}
\mathcal{N}\!\left( \sum_i p_i \ket{\psi_i}\!\bra{\psi_i} \right)
\neq
\sum_i p_i \, \mathcal{N}\!\left( \ket{\psi_i}\!\bra{\psi_i} \right).
\label{eq:nonlinearity}
\end{equation}
Equation~\eqref{eq:nonlinearity} expresses the fundamental failure of ensemble equivalence.
The map $\mathcal{N}$ depends on the particular ensemble decomposition of $\rho$ and therefore
does not define a linear map on density operators.

\end{document}